\documentclass[letter paper, 10 pt, conference]{ieeeconf}

\IEEEoverridecommandlockouts                              % This command is only needed if 
\usepackage[utf8]{inputenc}
\usepackage{amsmath, amssymb}
\usepackage{graphicx}
\graphicspath{{arXiv_Figures/}}
\usepackage{color}
\usepackage{ulem}
\usepackage{subcaption}
\usepackage{tikz}
\usepackage{alphabeta}
\usetikzlibrary{calc,intersections,through,backgrounds, math, angles, quotes}
\usepackage{tikzscale}
\tikzset{point/.style={circle, fill, inner sep=1.7}}
\usepackage{mathrsfs}
\usepackage{mathtools}
\usepackage{nicefrac}
\usepackage{cite}

\usepackage{cleveref}
\crefmultiformat{equation}{(#2#1#3)}%
   { and~(#2#1#3)}{, (#2#1#3)}{ and~(#2#1#3)}
\crefrangeformat{equation}{(#3#1#4) to~(#5#2#6)}
\crefformat{equation}{(#2#1#3)}
\crefmultiformat{figure}{Fig.~#2#1#3}%
{ and~Fig.~#2#1#3}{, Fig.~#2#1#3}{ and~Fig.~#2#1#3}
\crefrangeformat{figure}{Figs.~#3#1#4 to~#5#2#6}
\crefformat{figure}{Fig.~#2#1#3}

\renewcommand{\H}{\mathscr{H}}
\DeclareMathOperator*{\argmin}{\arg\min}

\DeclareMathOperator*{\atan}{atan2}

\usepackage{bm}

\usepackage{enumerate}

\newcommand{\x}{\mathbf{x}}

\newcommand{\rcap}{r_{\rm cap}}
\newcommand{\T}{^{\mbox{\tiny \sf T}}}

\newcommand{\trm}{t_{\rm r}}

\newcommand{\thetaEntry}{\theta_{\tan}}
\newcommand{\thetaExit}{\theta_{\rm exit}}
\newcommand{\psiExit}{\psi_{\rm exit}}
\newcommand{\phiExit}{\phi_{\rm exit}}
\newcommand{\tarc}{t_{\rm arc}}
\newcommand{\ts}{t_{\rm s}}

\newtheorem{theorem}{Theorem}
\newtheorem{lemma}{Lemma}

\newtheorem{remark}{Remark}

\newtheorem{proposition}{Proposition}

\renewcommand{\baselinestretch}{0.92}

\renewenvironment{proof}{\textit{Proof:}}{\hfill $\blacksquare$}

\title{\LARGE\bf  Optimal Evasion from a Sensing-Limited Pursuer}
\author{Dipankar Maity, Alexander Von Moll, Daigo Shishika, Michael Dorothy
\thanks{
We gratefully acknowledge the support of ARL grant ARL DCIST CRA W911NF-17-2-0181. The views expressed in this paper are those of the authors and do not reflect the official policy or position of the United States Government, Department of Defense, or its components.
}
\thanks{D. Maity is with the Department of Electrical and Computer Engineering, University of North Carolina at Charlotte,  NC, 28223, USA.
Email: 		{{dmaity@uncc.edu}}
}
\thanks{
A.~Von Moll is with the Control Science Center,  Air Force Research Laboratory, WPAFB, OH, 45433, USA. 
Email: { alexander.von\_moll@us.af.mil}
}
\thanks{
D. Shishika is with the Department of Mechanical Engineering, George Mason University, Fairfax, VA, 22030, USA. 
Email: 		{ dshishik@gmu.edu}
}
\thanks{
M. Dorothy is with the Army Research Directorate, DEVCOM Army Research Laboratory, APG, MD, 20783, USA. 
Email:  {michael.r.dorothy.civ@mail.mil}
 }
}

\makeatletter
\newcommand{\linebreakand}{%
  \end{@IEEEauthorhalign}
  \hfill\mbox{}\par\vspace{10pt}\hspace{20pt}
  \mbox{}\hfill\begin{@IEEEauthorhalign}
}
\makeatother

\begin{document}

\maketitle
\thispagestyle{empty}
\pagestyle{empty}

\begin{abstract}
This paper investigates a partial-information pursuit evasion game in which the Pursuer has a limited-range sensor to detect the Evader.
Given a fixed final time, we derive the optimal evasion strategy for the Evader to maximize its distance from the pursuer at the end. 
Our analysis reveals that in certain parametric regimes, the optimal Evasion strategy involves a `risky' maneuver, where the Evader's trajectory comes extremely close to the pursuer's sensing boundary before moving behind the Pursuer.
Additionally, we explore a special case in which the Pursuer can choose the final time. In this scenario, we determine a (Nash) equilibrium pair for both the final time and the evasion strategy.
% We focus on a single stage game between moments of intermittent long-range sensing, subject to the constraint of a radial short-range sensor. We fully characterize the Evader's best response to known interval duration as well as provide a Nash equilibrium in the case that the Pursuer is able to select the interval duration at the start of the game.
\end{abstract}

\section{Introduction}

Many variants of pursuit evasion games (PEGs) have been studied in the literature with particular emphasis on either the geometry of the game environment \cite{dorothy2021one} or the situations that involve multiple pursuers and/or evaders \cite{garcia2020optimal,pourghorban2022target}.
However, these works assume perfect state information and complete knowledge of the game parameters (i.e., payoff and the capability of the opposing player).
%
% Consideration of partial information is a significant challenge for PEGs, and the literature on this aspect remains sparse. 
%
Extensions of pursuit-evasion games in the context of limited perception have mainly considered limited sensing range e.g., \cite{durham2010distributed}, limited field of view e.g.,  \cite{gerkey2006visibility}, and noisy measurements e.g., \cite{bagchi1981linear}, but the challenges associated with the lack of continuous sensing remains unsolved.

This paper is motivated by the scenario involving a Pursuer ($P$) with limited perception capability.
$P$ is assumed to have two ways of sensing the Evader ($E$):
\begin{itemize}
\item Passive sensing: $E$ is sensed if it is within some sensing range from $P$.
\item Active sensing: $E$ is sensed anywhere if $P$ activates a costly measurement device.
\end{itemize}
Solving for the optimal/equilibrium strategies for this problem involves identification of the timing at which $P$ should activate the sensing (based on the type and amount of cost it must pay) as well as its heading when $E$ is invisible.

% \todo{Obstacle avoidance / constrained optimal control.}

Some of the closely related prior work includes \cite{aleem2015self, maity2016strategies, maity2016optimal, huang2021defending, maity2023efficient}, where the problem of active sensing is studied. 
Refs. \cite{maity2016optimal, maity2016strategies, huang2021defending, maity2023efficient} study a linear-quadratic game formulation and do not consider the passive sensing aspect of $P$.
However, these works have studied the optimal evader strategy between two sensing instances. 
The passive sensing capability of $P$ makes the problem more challenging, and the derived strategies in \cite{maity2016optimal, maity2016strategies, huang2021defending, maity2023efficient} are not necessarily optimal and sometimes infeasible in presence of passive sensing.
Ref. \cite{aleem2015self} does not investigate the optimal evasion strategy; instead, the focus was on designing active sensing instances to ensure the eventual capture of $E$.

The work presented in this paper solves a subproblem of the scenario introduced above.
Specifically, we extract the time period between two consecutive moments of active sensing by $P$ and assume that $P$ will move in a straight line trajectory during this period. 
We then proceed to solve two problem variations for this interval of time:
\begin{enumerate}
    \item \textbf{Optimal Control Problem:} Given a fixed interval of time, what is the best-response strategy of $E$  to maximize its final distance from $P$?
    \item \textbf{Nash Game:} If $P$ can select the duration of the time interval (i.e., the intermittent sensing times), what is the equilibrium strategy pair such that neither player can benefit from unilateral deviation?
\end{enumerate}
% By fixing the Pursuer strategy in (1), the game reduces to an optimal control problem for the Evader.
% We identify the optimal Evader strategy such that it stays outside the passive sensing range for the whole time, and the distance from the Pursuer is maximized at the final time (i.e., when the next active sensing happens).

\section{Problem Formulation} \label{sec:formulation}
% Let us consider $P$ located at the origin of an inertial frame.
%our coordinate system $\x_P(0) = \begin{bmatrix}0 & 0\end{bmatrix}\T$.\ds{Can we present this in a way that this is okay wlog?  Is it obvious?}
$P$ has a capture radius of $\rcap$ and it is moving along the positive $x$-axis with a constant speed of $v_P$. 
Throughout this paper, we will refer to the circle around $P$ with radius $\rcap$ as the \textit{Proximity Circle}. 
% The Pursuer does not have any sensing capability and move along the $x$-axis with speed $v_P$ for a duration of $T$. 
$E$ starts outside the Proximity Circle and has an objective to avoid getting captured by $P$ and maximize its final distance from $P$ at final time $t=T$. 
In case capture is unavoidable, the objective of $E$ is to delay capture as long as possible, i.e., maximize its \textit{survival time}. 

The dynamics of $P$ and $E$ can be expressed as
\begin{align}
    \dot\x_P(t) = v_P\begin{bmatrix}
        1 \\ 0
    \end{bmatrix}, \qquad 
    \dot\x_E(t) = v_E \begin{bmatrix}
         \cos\psi(t) \\  \sin \psi(t)
    \end{bmatrix},
\end{align}
where $v_E < v_P$ denotes the constant speed of $E$ and $\psi(t)$ denotes the instantaneous heading angle of $E$.
Without loss of generality, one may assume $v_P = 1$ and $v_E = \mu < 1$ and $P$'s capture radius $\rcap=1$.\footnote{
Define the new variables $\bar\x_i(t) = \frac{1}{\rcap}\x_i \big(\rcap\frac{t}{v_P}\big)$ for $i = P, E$ and $\mu = v_E/v_P$. Then, one obtains $\dot{\bar\x}_P = \begin{bmatrix}
    -1 & 0
\end{bmatrix}\T$, and $\dot{\bar\x}_E = \mu \begin{bmatrix}
         \cos\psi(t) &  \sin \psi(t)
    \end{bmatrix}\T$.
    In this transformed coordinate system, the Pursuer's capture radius becomes $1$.
}

To proceed with our analysis, let us describe $E$'s dynamics in the Pursuer-fixed frame:
\begin{equation}
    \label{eq:f_nondim}
   \dot{\x}(t) =  
    \begin{bmatrix}
        \mu \cos{\psi(t)} - 1 \\
        \mu \sin{\psi(t)}
    \end{bmatrix}, \quad
    \x(0)=\x_0,
\end{equation}
where $\x = \x_E - \x_P = \begin{bmatrix} x & y \end{bmatrix}\T$.
The capture avoidance constraint is described as
\begin{equation}
    \label{eq:c_nondim}
    S(\x) = \|\x\| - 1 \ge 0,
\end{equation}
for all $t \in \left[ 0, T \right]$.
Due to symmetry across the $x$-axis, we assume that $y\ge0$ without loss of generality.

The objective cost functional, which $E$ wishes to maximize, is
\begin{equation}
    \label{eq:J}
    J = Φ(\x_{f}) = \|\x_{f}\|^2,
\end{equation}
where the subscript $f$ denote conditions at final time (i.e., when $t = t_f = T$).
The terminal manifold is given by the zero-level set of the function
\begin{equation}
    \label{eq:ϕ}
    F(\x_{f}, t_f) = t_f - T.
\end{equation}
In case capture is unavoidable (i.e., for every possible evading strategy, there exists a $t\le T$ such that $S(\x(t)) < 0$), $E$ wishes to maximize $t_{\rm{cap}}$, where $t_{\rm{cap}}$ is the earliest time such that $S(\x(t_{\rm{cap}})) = 0$.

%%%%%%%%%%%%%%%%%%%%%%%%%%%%%%%%%%%%%%%%%%%%%%%%%%%%%%%%%%%%%%%%%%%%%%%%

\section{Some Geometric Relationships}\label{sec:geometry}

% \begin{figure}
%     \centering
%     \includegraphics[width=0.5\linewidth]{arXiv_Figures/schematic.png}
%     \caption{Enter Caption}
%     \label{fig:enter-label}
% \end{figure}

\begin{figure}
    \centering
    \includegraphics[width=0.8\linewidth]{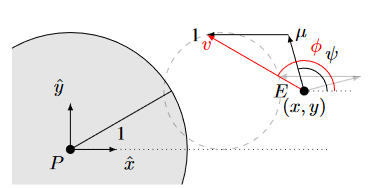}
    \caption{Constrained max distance schematic. The quantities $ϕ$ and $v$ represent the Evader's effective heading and velocity in the Pursuer-fixed frame. 
    The grey vectors demonstrate that there is another heading, $ψ$, which results in the same $ϕ$ but with a reduced effective velocity. 
    The dashed gray circle is the locus of the end points of the resultant velocity vector $v$. 
    }
    \label{fig:constrained_max_distance_schematic}
\end{figure}
%In the Pursuer-fixed frame, the instantaneous velocity of $E$ is $\begin{bmatrix}(\mu \cos{\psi} - 1)& \mu \sin{\psi}\end{bmatrix}\T$, 
%    as given in \eqref{eq:f_nondim}.
We begin by noting some relationships between E's choice of velocity in the inertial frame and the correponsing velocity in the Pursuer-fixed frame.
Let $v$ and $\phi$ denote the speed and instantaneous heading angle of $E$ in the Pursuer-fixed frame (c.f., Fig.~\ref{fig:constrained_max_distance_schematic}). 
The relationship between $\phi, \psi$ and $v$ is
\begin{align}
\label{eq:v_mu_phi_psi}
\begin{split}
        &v \sin\phi = \mu \sin\psi, \\
        &v \cos\phi = -1 + \mu \cos\psi.    
\end{split}
\end{align}
Given that speed $v$ is non-negative, we immediately obtain $\cos\phi <0$, or $|\phi| > \pi/2$. 
One may further verify that $\phi$ belongs in the range $[\pi-\sin^{-1}\mu,~ \pi + \sin^{-1} \mu]$ for every possible choice of $\psi$.
One may also verify that $v$ is in the range $[1-\mu,~ 1+\mu]$.  

Equation \eqref{eq:v_mu_phi_psi} can be used to solve for two quantities (e.g., $v$ and $\phi$) given the third one. 
For a given $\psi$, we obtain 
\begin{align}
\label{eq:given_psi}
\begin{split}
    &\phi(\psi) = \atan \left({\mu \sin\psi}, ~{\mu\cos\psi - 1}\right)\\
    &v(\psi) = \sqrt{1 + \mu^2 - 2\mu\cos ψ}.
    \end{split}
\end{align}
For a given $\phi \in [\pi-\sin^{-1}\mu,~ \pi + \sin^{-1} \mu]$, we obtain
\begin{align}
\label{eq:given_phi}
    &\big(\psi(\phi), ~v(\phi)\big) \nonumber \\
    &= \begin{cases}
        \left( ϕ - \sin^{-1} \left( \frac{\sin ϕ}{\mu}\right),~~ -\cos\phi + \sqrt{\mu^2 - \sin^2 \phi}\right). \\
        \left( ϕ + \sin^{-1} \left( \frac{\sin ϕ}{\mu}\right) - \pi, ~ -\cos\phi - \sqrt{\mu^2 - \sin^2 \phi} \right).
    \end{cases}
\end{align}
% For a given speed $v \in [(1-\mu), ~(1+\mu)]$, we obtain
% \begin{align}
% \label{eq:given_v}
%     &\big(\psi(v), ~\phi(v)\big) = \nonumber \\
%     &\pm \bigg(\!\! \cos^{-1} \!\!\bigg( \frac{1+\mu^2 - v^2}{2\mu} \bigg), \pi - \cos^{-1}\!\!\bigg( \frac{1 + \mu^2 - v^2}{2\mu} \bigg)\!\! \bigg). 
% \end{align}
Notice from \eqref{eq:given_phi} that  %(or $v$)
there exist two solutions for a given $\phi$ (c.f., \Cref{fig:constrained_max_distance_schematic}) .

\section{Regions of Interest and Max Survival Time}

\begin{lemma}\label{lm:initConInABand}
    For all states $\|\x\|>1$ in the Pursuer-fixed frame such that $x\leq0$ or $y\geq1$, the  optimal trajectory starting at $\x$ avoids the constraint.
\end{lemma}
\begin{proof}
    This is true by inspection, as $E$ is maximizing its final distance from the origin of the Pursuer-fixed frame.
\end{proof}

The initial conditions of interest for the remainder of this paper are those where Lemma~\ref{lm:initConInABand} does \textit{not} hold.

When $E$ is on the Proximity Circle, its position is uniquely characterized by the angle $\theta = \atan(y, x)$. 
In subsequent sections, when we say that $E$ is at angle $\theta$ on the Proximity Circle, we imply $\x = [\cos\theta ~~ \sin\theta]\T$. 

\begin{lemma}\label{lm:noEscape}
    At any time $t$, if $E$ is on the Proximity Circle with $|\theta| < \cos^{-1}\mu$, then $\|\x(t^+)\| < 1$ and consequently, capture happens.
\end{lemma}
\begin{proof}
    %Transform the state to polar coordinates, $(r,\theta)$, where on the Proximity Circle, $r=1$. Then, $\dot{r}=-\cos(\theta)+\mu\cos(\theta-\phi)$. When $\theta < \cos^{-1}\mu$, $\dot{r}<0$.  
%
    Denote $d(t) = \|\x(t)\|$. Then,
    \begin{align*}
        \dot d = \frac{\x\T \dot \x}{\|\x\|}.
    \end{align*}
    Using \Cref{eq:f_nondim} and \Cref{eq:v_mu_phi_psi} to substitute $\dot \x = [v\cos\phi ~~ v\sin\phi]\T$ and given that $\x(t)=[\cos\theta~~\sin\theta]\T$, we obtain
    \begin{align*}
        \dot d = v \cos(\phi - \theta).
    \end{align*}
    Since $\phi \in [\pi-\sin^{-1}\mu,~ \pi + \sin^{-1} \mu]$ and $\theta \in (-\cos^{-1}\mu, \cos^{-1} \mu)$, we have $|\phi-\theta| > \pi - \sin^{-1}\mu - \cos^{-1}\mu = \frac{\pi}{2}$.
    Consequently, $\dot d < 0$ and $d(t^+) <  1$.
\end{proof}

\Cref{lm:noEscape} implies that $E$ must avoid ending up on the Proximity Circle with $|\theta|< \cos^{-1}\mu$.
It is noteworthy that, for some initial conditions $\x_0$ and game duration $T$, capture is inevitable. 
The next lemma provides the necessary and sufficient conditions on $\x_0$ and $T$ for $E$ to end up on the Proximity Circle at a $\theta \in (-\cos^{-1}\mu, \cos^{-1} \mu)$. 
 \begin{lemma} \label{lem:guaranteedCapture}
 Capture is guaranteed iff $E$'s initial location $\x_0 \equiv \begin{bmatrix}x_0 & y_0\end{bmatrix}\T$ and the game duration $T$ satisfy 
    \begin{align} \label{eq:no_escape_zone}
        \mu x_0 + \sqrt{1- \mu^2}y_0 < 1,\quad x_0 > \mu,
    \end{align}
    \begin{align}\label{eq:max_survival_time}
    \begin{split}
    T > T_{\rm{survive}}(\x_0),
    \end{split}
\end{align}
where 
\begin{align*}
    T_{\rm{survive}}(\x_0) = \tfrac{1}{1-\mu^2}\left[ (x_0-\mu) - \sqrt{(1-\mu x_0)^2 - (1-\mu^2)y^2_0} \right].
\end{align*}
\end{lemma}

\begin{proof}
    The proof is presented in \Cref{AP:guaranteedCapture}.
    % [sketch] For any $\x_0$ satisfying \eqref{eq:no_escape_zone} and $T$ satisfying \eqref{eq:max_survival_time}, $E$'s trajectory will end up on the Proximity Circle with $\theta< \cos^{-1}\mu$. 
\end{proof}

% \ds{Remark/discussion about geometric intuition.}
%
For a given $T$, the set of initial locations for which $E$ is \textit{not} able to avoid capture is given by the set:
\begin{align*}
    \begin{multlined}
    \Omega(T) = \left\{\x_0~|~ (x_0 - T)^2 + y_0^2 < (1-\mu T)^2, \right. \\
    \left. \text{and  \eqref{eq:no_escape_zone} is satisfied}\right\}.
    \end{multlined}
\end{align*}
The set $\Omega_{\text{capture}} = \cup_{T\ge 0} \Omega_T$ is called the ``no-escape zone" in this paper. 
For any $\x_0 \in \Omega_{\text{capture}}$ there exists a final time $T$, such that capture is inevitable. 
For a given $\x_0 \in \Omega_{\text{capture}}$, the maximum \textit{survival time} is $T_{\rm{survive}}(\x_0)$.
An illustration is provided in \Cref{fig:survival_time}.
\begin{figure}
    \centering
    \includegraphics[trim = 135 320 160 340, clip, width = 0.9 \linewidth]{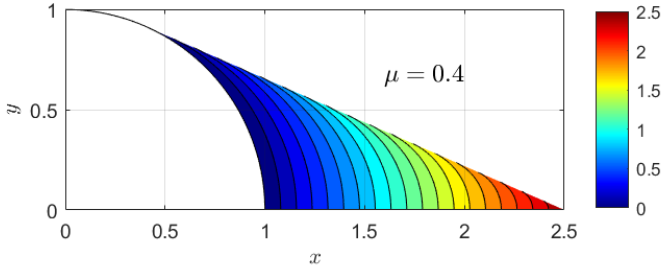}
    \caption{Level curves of $T_{\rm{survive}}$ for $\x_0 \in \Omega_{\text{capture}}$.}
    \label{fig:survival_time}
\end{figure}

\begin{remark}\label{rem:nashUnconstrainted}
    If $\x_0 \notin \Omega_{\text{capture}}$, then there always exists a strategy to avoid capture. 
For instance, if $E$ moves along $\psi = \cos^{-1}\mu$, then $\|\x(t)\| \ge 1$ for all $t$. 
To verify this, we observe from \eqref{eq:f_nondim} that $\x(t) = \x_0 + t \begin{bmatrix} (\mu^2 -1) & \mu\sqrt{1-\mu^2} \end{bmatrix}\T$, and consequently, $\|\x(t)\|^2 \ge (\mu x_0 + \sqrt{1- \mu^2}y_0)^2 \ge 1$, for all $t$.
Such a heading angle is not necessarily optimal.
\end{remark}

In subsequent sections, we analyze the optimal strategy for $E$ with the initial conditions $\x_0 \notin \Omega_{\text{capture}}$.

%%%%%%%%%%%%%%%%%%%%%%%%%%%%%%%%%%%%%%%%%%%%%%%%%%%%%%%%%%%%%%%%%%%%%%%%%%%%%%

\section{Unconstrained Optimal Control Analysis} \label{sec:optimal_control}

\begin{lemma} \label{lem:optimal_control}
    The optimal Evader heading $ψ^*$ is constant backwards in time until either 1) the constraint is reached ($S = 0$) or 2) the initial condition is reached ($\mathbf{x}_0 = \begin{bmatrix}x_0 & y_0\end{bmatrix}\T$).
\end{lemma}

\begin{proof}
The Hamiltonian for the system is
\begin{equation}
    \label{eq:H_generic}
    \H = \dot{\mathbf{x}}\T \bm{λ} = λ_x \left( \mu\cos ψ - 1 \right) + λ_y \mu \sin ψ,
\end{equation}
where $\bm{λ} = \begin{bmatrix} λ_x & λ_y\end{bmatrix}\T$ is the adjoint vector.
The terminal adjoint values are given by~\cite{bryson1975applied}
\begin{equation}
    \label{eq:λf}
    \bm{λ}_f\T = \frac{\partial Φ}{\partial \mathbf{x}_{f}} + γ \frac{\partial F}{\partial \mathbf{x}_{f}} = \begin{bmatrix} 2x_f & 2y_f\end{bmatrix},
\end{equation}
where $γ$ is an additional adjoint variable and $\x_{f} \equiv \begin{bmatrix}x_f & y_f\end{bmatrix}\T$ is the terminal state.
The optimal adjoint dynamics are given by~\cite{bryson1975applied}
\begin{equation}
    \label{eq:λdot}
    \dot{\bm{λ}} = -\frac{\partial \H}{\partial \mathbf{x}} = \mathbf{0},
\end{equation}
since the state does not appear in the Hamiltonian.
Therefore, $\bm{λ}(t) = \bm{λ}_f$ for all $t< t_f$ such that the constraint $S=0$ has not been activated.
% The terminal Hamiltonian value is~\cite{bryson1975applied}
% %
% \begin{equation}
%     \label{eq:H_f}
%     \H_f = -\frac{\partial Φ}{\partial t_f} - γ \frac{\partial F}{\partial t_f} = -γ.
% \end{equation}
%
Evaluating~\Cref{eq:H_generic} at $t = t_f$ and substituting in the terminal adjoint values, \Cref{eq:λf}, gives
\begin{equation}
    \label{eq:H_f_substituted}
    \H_f  = 2x_f\left( \mu\cos ψ_f - 1 \right) + 2y_f \mu\sin ψ_f .
\end{equation}
The optimal Evader heading at final time must maximize this terminal Hamiltonian, which implies
\begin{equation}
    \label{eq:ψf}
    \cos ψ_f^* = \frac{x_f}{\sqrt{x_f^2 + y_f^2}}, \qquad
    \sin ψ_f^* = \frac{y_f}{\sqrt{x_f^2 + y_f^2}},
\end{equation}
i.e., the terminal heading is aligned with the vector from the origin to the Evader's terminal position.
Similarly, for all $t < t_f$ the heading $ψ_f^*$ maximizes $\H$ in~\Cref{eq:H_generic}, which implies the optimal heading $ψ^*$ is constant and equal to $ψ^*_f$ for all $t\le t_f$ such that the constraint is not activated.
% 
% Therefore, the optimal Evader heading $ψ^*$ is constant and equal to $ψ^*_f$ from time $t_f$ (backwards in time) until either 1) the constraint is reached ($S = 0$) or 2) the initial condition is reached ($\mathbf{x_0} = \begin{bmatrix}x_0 & y_0\end{bmatrix}\T$).
\end{proof}

For a constant heading, $ψ$, and first assuming that the constraint $S = 0$ does not become active, the system~\Cref{eq:f_nondim} can be solved:
\begin{equation}
    \label{eq:xy_constantψ}
    \begin{aligned}
        \x(t) = \begin{bmatrix}
            x_1 + t \left( \mu \cos ψ - 1 \right)\\
            y_1 + t \mu\sin ψ
        \end{bmatrix}
    \end{aligned}
\end{equation}
for $t \in \left[ t_1, T \right]$ where $\x(t_1) \equiv \begin{bmatrix}x_1 & y_1\end{bmatrix}\T$.
Thus the terminal position is given by
\begin{equation}
    \label{eq:xyf}
    \begin{bmatrix} x_f \\ y_f\end{bmatrix} =
    \begin{bmatrix}
        x_1 + (T - t_1) \left( \mu \cos ψ - 1 \right) \\
        y_1 + (T - t_1) \mu \sin ψ 
    \end{bmatrix}.
\end{equation}
Combining \Cref{eq:ψf} and \Cref{eq:xyf}, the optimal heading can be rewritten as
\begin{align}
    \label{eq:ψ_xy}
\begin{split}
    \cos ψ^* = \frac{x_1 - (T - t_1)}{\|\x_{f}\| - \mu(T - t_1)},\\
    \sin ψ^* = \frac{y_1}{\|\x_{f}\| - \mu(T - t_1)},
\end{split}
\end{align}  
or, more compactly and generally, as
\begin{equation}
    \label{eq:ψ_xy_compact}
    ψ^*(t,\x) = \atan\left( y, ~x - (T - t) \right),
\end{equation}
where $\begin{bmatrix}x & y\end{bmatrix}\T = \x(t)$.

\begin{figure}
    \centering
    \includegraphics[width=0.8\linewidth]{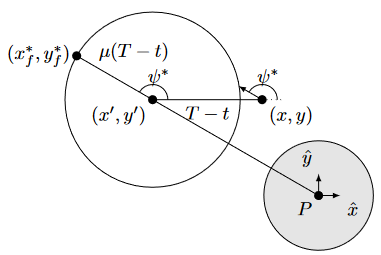}
    \caption{Geometric interpretation of the optimal unconstrained heading.}
    \label{fig:unconstrained-heading}
\end{figure}

This result also comes about via an intuitive geometric interpretation.
Suppose at time $t$, $E$'s position in the Pursuer-fixed frame is $\x(t) = \begin{bmatrix}     x & y \end{bmatrix}\T$.
The drift due to $P$'s motion can be applied for the remaining time until $T$, giving a virtual point $ \begin{bmatrix}     x' & y' \end{bmatrix}\T = \begin{bmatrix}     (x-(T-t)) & y \end{bmatrix}\T$.
Then, one may consider the Evader's reachability set starting from this virtual point, which is a circle centered at that virtual point with radius $\mu(T - t)$.
By inspection, the point on the circle furthest from $P$ lies on the line passing from the origin the Pursuer-fixed frame through $\begin{bmatrix}x'& y'\end{bmatrix}\T$, and thus the optimal heading is the one given in~\Cref{eq:ψ_xy_compact} (see~\Cref{fig:unconstrained-heading}).

% \dm{Perhaps one sentence is needed to summarize the importance of this section.}
The result in this section provides the unconstrained optimal trajectory, and it will be the full solution for the case where the entire trajectory does not intersect with the interior of the Proximity Circle (i.e., $S(\x(t))$ is always non-negative.). Section~\ref{sec:constrained_trajectories} will show that this trajectory is valid (with slight modification) to describe the final portion of the constrained trajectory, after the constraint has been eclipsed.

\section{Trajectories where the Constraint is Activated}
\label{sec:constrained_trajectories}

We now turn to the case where the constraint necessarily becomes active at some time. 
\Cref{lem:optimal_control} in \Cref{sec:optimal_control} states that the last portion of the optimal trajectory will be a straight line segment that exits from the Proximity Circle at some time. 
Therefore, unless $\|\x_0\|=1$, the optimal trajectory will enter the Proximity Circle, then ride along this circle for some period of time, and then leave the constraint. 
This section will consider each of these three phases and show that (i) There is a unique optimal trajectory for the first unconstrained phase (entering the constraint) that is independent of $T$, (ii) The time spent on the constraint can be described by the solution to an elliptic integral, and (iii) The optimal exit time is the first moment that the (modified) unconstrained optimal control from Section~\ref{sec:optimal_control} is feasible, and the trajectory remains unconstrained thereafter.

The following lemma states the necessary and sufficient condition for which a straight line trajectory will intersect the Proximity Circle.
%
% \begin{lemma}\label{lm:constraintActive}
%     For given initial location $\x_0$, the straight line trajectory with heading $\psi^* = \atan(y_E,\ x_E-T)$ intersects with the Proximity Circle iff $\phi_{\tan} + \atan  \left(\mu \sin ψ^*,\ 1 - \mu \cos ψ^* \right)< \pi$ and $T > t_c$, where $\phi^*$ is obtained from \eqref{eq:given_psi} by plugging in $\psi = \psi^*$, and 
%     \begin{equation}
%     \label{eq:t_c}
%     \begin{multlined}
%         t_c = \frac{1}{v(ψ^*)}\left[-(x_E\cos\phi^* + y_E\sin\phi^*) - \right. \\
%         \left. \sqrt{1 - (x_E\sin\phi^* - y_E\cos\phi^*)^2} \right]
%         \end{multlined}
%     \end{equation}
%     is the amount of time required to collide with the Proximity Circle, where $v(\psi^*) = \sqrt{1 + \mu^2 - 2\mu\cos ψ^*}$ is obtained from \eqref{eq:given_psi}.
% \end{lemma}
%
% A general psi result.
\begin{lemma}\label{lm:constraintActive}
    For given initial location $\x_0$, a straight line trajectory with heading $\psi$ intersects with the Proximity Circle iff $\phi_{\tan} + \atan  \left(\mu \sin ψ,\ 1 - \mu \cos ψ \right)< \pi$ and $T > t_c$, where  
    \begin{equation} \label{eq:phi_tan}
    ϕ_{\tan} = \pi - \sin^{-1} \left( \tfrac{1}{\|\x_0\|} \right) + \atan \left( y_0,\ x_0 \right),
\end{equation}
 and
    \begin{equation}
    \label{eq:t_c}
    \begin{multlined}
        t_c = \frac{1}{v(ψ)}\left[-(x_0\cos\phi + y_0\sin\phi) - \right. \\
        \left. \sqrt{1 - (x_0\sin\phi - y_0\cos\phi)^2} \right]
        \end{multlined}
    \end{equation}
    is the amount of time required to collide with the Proximity Circle, $v(\psi)$ and $\phi$ are obtained from \eqref{eq:given_psi}.
\end{lemma}
\begin{proof}
    One may verify that (c.f., \Cref{fig:tanEntry}) the angle for which the effective heading of $E$ in the Pursuer-fixed frame is tangent to the Proximity Circle is given by $ϕ_{\tan}$.
    On the other hand, for a given $\psi$, the effective heading angle in Pursuer-fixed frame is $\pi - \atan \left(μ \sin ψ,\ 1 - μ \cos ψ \right)$, as given by \eqref{eq:given_psi}.
    Firstly, if $ϕ_{\tan} + \atan \left(μ \sin ψ,\ 1 - μ \cos ψ \right) \ge \pi$ then the heading $ψ$ results in $ϕ$ that points `outside' the Proximity Circle and thus $E$ will not enter it.
    Secondly, if $ψ$ results in $ϕ$ that points `inside' the Proximity Circle and $T$ is less than the time it takes for $E$ to enter the Proximity Circle under $ψ$ then $E$ will not enter it.
    The collision time $t_c$ in~\Cref{eq:t_c} is obtained by dividing the distance of $E$ from the Proximity Circle along the trajectory associated with $ψ$ by its velocity along this trajectory, $v(ψ)$.
    The bracketed expression in~\Cref{eq:t_c} represents the aforementioned distance and may be obtained via the Law of Cosines (c.f.~\Cref{fig:tanEntry} but for general $ϕ(ψ)$).
\end{proof}
\begin{figure}
    \centering
    \includegraphics[width=0.8\linewidth]{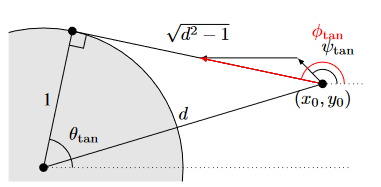}
    \caption{Geometric relationship of $ϕ_{\tan}$ and the Proximity Circle entry point defined by $θ_{\tan}$.}
    \label{fig:tanEntry}
\end{figure}

Lemma~\ref{lm:constraintActive} implies that if the optimal heading angle $\psi^*$ from \eqref{eq:ψ_xy_compact} is such that either $T<t_c$ or $\phi_{\tan} + \atan  \left( \mu \sin ψ^*,\ 1 - \mu \cos ψ^* \right)\geq \pi$, then the entire straight line trajectory with heading $\psi^*$ is unconstrained, and the analysis from Section~\ref{sec:optimal_control} determines the full solution. Otherwise, the constraint must become active at some time.
In \Cref{fig:constrainted_activated_region} we divide the initial locations $\x_0$ into regions depending on whether the constraint becomes active or not.
The following subsections analyze each phase of the constrained trajectory.
\begin{figure}
    \centering
    \includegraphics[trim = 155 320 150 350, clip, width = \linewidth]{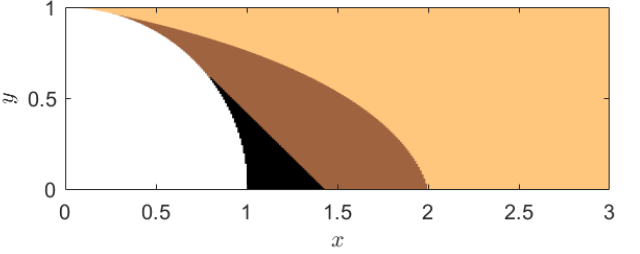}
    \caption{The initial positions of $E$ are divided into three regions depending on whether $E$ will be able escape $P$ and whether the escape trajectory will activate the proximity constraint. \colorbox{black}{O}: no-escape zone $\Omega_{\rm{capture}}$, \colorbox{brown}{\textcolor{brown}{O}}: Constraint is activated, \colorbox{orange!45}{\textcolor{orange!45}{O}}: optimal trajectories are unconstrained. For this example $T=2$ and $\mu = 0.7$.
    }
    \label{fig:constrainted_activated_region}
\end{figure}

\subsection{Phase I: Entering the Constraint}\label{sec:enteringConstraint}

% The angle for which the effective heading of $E$ in the Pursuer-fixed frame is tangent to the Proximity Circle is given by
% %
% \begin{equation} \label{eq:phi_tan}
%     ϕ_{\tan} = \pi - \sin^{-1} \left( \tfrac{1}{\|\x_0\|} \right) + \atan \left( y_0,\ x_0 \right),
% \end{equation}
% as shown in~\Cref{fig:tanEntry}.
%

From \eqref{eq:given_phi} we notice that two different heading angles $\psi_{\tan}$ (i.e., $\psi_{\tan} = ϕ_{\tan} - \sin^{-1} \left( \nicefrac{\sin ϕ_{\tan}}{\mu}\right)$ or $\psi_{\tan} = ϕ_{\tan} + \sin^{-1} \left( \nicefrac{\sin ϕ_{\tan}}{\mu}\right) - \pi$) of $E$ lead to the same $\phi_{\tan}$. 
The heading angle $\psi_{\tan} = ϕ_{\tan} - \sin^{-1} \left( \nicefrac{\sin ϕ_{\tan}}{\mu}\right)$ requires a shorter time to reach the tangent point since the speed $v$ along this heading is higher. 
Therefore, the (shortest) time required to reach the tangent point from $\x_0$ is
\begin{equation}\label{eq:t_tan}
    t_{\tan} = \frac{\sqrt{\|\x_0\|^2 - 1}}{-\cos\phi_{\tan} + \sqrt{\mu^2 - \sin^2 \phi_{\tan}}}.
\end{equation}

% Define the angle $θ = \atan (y,\ x)$ which is the Evader's position on the Proximity Circle, as shown in~\Cref{fig:tanEntry}. 
% When $E$ is on the Proximity Circle, denote $\x(t) = \begin{bmatrix}\cos\theta(t) & \sin\theta(t)\end{bmatrix}\T$ as $E$'s position at time $t$.
At time $t_{\tan}$, $E$ will be at an angle $\theta_{\tan}$ on the Proximity Circle, as shown in \Cref{fig:tanEntry}, where
\begin{align}\label{eq:thetaTan}
    \theta_{\tan} = \cos^{-1}\left( \tfrac{1}{\|\x_0\|}\right) + \atan \left( y_0,\ x_0 \right). 
\end{align}

\begin{remark}\label{rm:thetaEntry}
    $t_{\tan}, \phi_{\tan}$ and $\theta_{\tan}$ depend on the initial location of $E$ $\x_0$ and not on $T$ or $\mu$. 
    Moreover, from Lemma~\ref{lm:initConInABand}, all constrained trajectories have $\thetaEntry\in[0,\frac{\pi}{2})$.
\end{remark}   

Consider the feasible entry location $\theta_{\tan}$ from~\eqref{eq:thetaTan}, which is independent of $T$. We will now show that, given that the constraint cannot be avoided entirely, $\thetaEntry$ is the best entry location to the Proximity Circle to reach any $\theta > \thetaEntry$ on the Proximity Circle.

\begin{theorem}\label{th:thetaEntry}
    The optimal Evader trajectory to reach a location on the constraint $\theta\geq\theta_{\tan}$ is to follow a straight line trajectory to $\theta_{\tan}$ and then proceed along the constraint to $\theta$.
\end{theorem}

 \begin{proof} 
 % The time spent to reach $\theta$ is
 %    \begin{align} \label{eq:t_s}
 %        \ts(\theta) = t_{\tan} + \tarc(\thetaEntry, \theta),
 %    \end{align}
 %    where $t_{\tan}$ and $\tarc$ are given in \eqref{eq:t_tan} and \eqref{eq:t_arc}, respectively. 
Any straight line trajectory to a $\theta > \thetaEntry$ will hit the Proximity Circle at an angle $\theta' < \thetaEntry$. 
 Therefore, $E$ will be forced to ride along the Proximity Circle from $\theta'$ to $\theta$, while passing $\thetaEntry$ on the way. 
 The theorem is proved by showing that the time optimal trajectory to reach $\thetaEntry$ is by following the heading $\phi_{\tan}$ (c.f. \eqref{eq:phi_tan}) in the Pursuer-fixed frame. 
 To that end, we note that for any trajectory in which $E$ enters the constraint at some $\theta < \theta_{\tan}$ and then remains on the constraint until $\theta = \theta_{\tan}$ is suboptimal since the straight line trajectory to $\thetaEntry$ is feasible~\cite{zermelo1931navigationsproblem}.
 \end{proof}

\subsection{Phase II: Riding the Constraint}

% In order to remain on the constraint, the Evader's effective velocity must be tangent to its position on the Proximity Circle: $ϕ = θ + \tfrac{\pi}{2}$.
% From Section~\ref{sec:geometry}, for this effective heading to be physically realizable, it must be in the range $\left[\pi - \sin^{-1} \mu, \pi + \sin^{-1} \mu \right]$.
% The critical case occurs when $ϕ = ϕ_c$ (i.e., the boundary of this range) whereby the critical position is $θ_c = \pm \cos^{-1} \mu$.
% Thus for positions $θ \in \left(-\cos^{-1}\mu, \cos^{-1} \mu \right)$ there does not exist an Evader control which can prevent violation of the constraint. 
% The Evader's optimal strategy to maximize survival time in this case will be discussed in Section~\ref{sec:maxSurvive}
%
For values of $\theta$ which admit survival on the constraint, the effective heading ($\phi$) of $E$ which keeps the constraint active satisfies $\cos ϕ = -\sin θ$ and $\sin ϕ = \cos θ$.
This relationship can be substituted into~\Cref{eq:given_phi} to obtain the angular velocity of $E$ along the constraint
\begin{equation}
    \dot{θ} = v(\phi) = \sqrt{\mu^2 - \cos^2 θ} + \sin θ.
\end{equation}
This expression can be rearranged and integrated to solve for the time spent on the constraint as a function of a starting ($\theta_1$) and ending ($\theta_2$) angular position:
\begin{align}
\begin{split} \label{eq:t_arc}
    \tarc(\theta_1, \theta_2) =& \frac{1}{1 - \mu^2} \left( \cos θ_1 - \cos θ_2 \right) \\
    &- \frac{\mu}{1 - \mu^2} \int_{θ_1}^{θ_2} \sqrt{1 - \frac{\cos^2 θ}{\mu^2}} \mathop{\mathrm{d}θ},        
\end{split}
\end{align}
where $\pi > |\theta_1| \ge \cos^{-1}\mu$.
The second term is an elliptic integral of the second kind and does not have a closed-form solution.

\subsection{Phase III: Exiting the Constraint}
%Assuming that the Evader took the optimal tangential trajectory to reach a location $\thetaEntry$, the remaining game duration is 
%This duration will be helpful in determining the optimal location to exit the constrained surface.
From Lemma~\ref{lm:noEscape} and Remark~\ref{rm:thetaEntry}, the constrained portion of any trajectory begins at $\thetaEntry\in[\cos^{-1} μ,\frac{\pi}{2})$ to survive the constraint. From Lemma~\ref{lm:initConInABand}, exit occurs no later than $\theta=\frac{\pi}{2}$.
%Assuming the game has not yet terminated, $\cos(\thetaEntry),\sin(\thetaEntry),t_r>0$.
Let us denote
\begin{align} \label{eq:t_remain}
    \trm(\theta) = T - \ts(\theta)
\end{align}
to be the remaining duration of the game, where $\ts(\theta)$ is defined as the time spent to reach location $\theta$ on the Proximity Circle:
\begin{align} \label{eq:t_s}
        \ts(\theta) = t_{\tan} + \tarc(\thetaEntry, \theta),
    \end{align}
 
If $E$ wishes to exit the Proximity Circle at $\trm$, then according to \eqref{eq:ψ_xy_compact},  $E$'s optimal heading in the global frame is
\begin{align}\label{eq:finalUnconstrained}
    \psi = \atan(\sin\theta,\ \cos\theta - \trm(\theta)). 
\end{align}
% In the Pursuer-fixed frame, this heading angle yields $\phi$ which satisfies \eqref{eq:given_psi} for the given $\psi$ above. 
\begin{theorem}  \label{thm:exit_angle}
    The optimal exit angle ($\thetaExit$) satisfies
    \begin{align} \label{eq:exitCondition}
        \atan(\mu\sin\psiExit,\ \mu\cos\psiExit - 1) - \thetaExit = \tfrac{\pi}{2}, 
    \end{align}
    where $\psiExit$ is obtained from \eqref{eq:finalUnconstrained} by substituting $\thetaExit$ for $\theta$.
    That is, at the optimal exit angle, the optimal unconstrained trajectory is tangent to the Proximity Circle.
\end{theorem}

\begin{proof} 
    %\avm{reword the beginning of this proof, Alex.}
    %The theorem is proved using two results: (i) there exists a unique $\thetaExit \in [\thetaEntry, \tfrac{\pi}{2}]$ that solves \eqref{eq:exitCondition}, and (ii)  the optimal action at $\thetaExit$ is to leave the Proximity Circle.
    %
    %Here we prove (ii) \dm{make a statement about how to prove (i)}. 
    Consider the point $\x_{\rm{exit}} = [\cos\thetaExit~~\sin\thetaExit]\T$ on the Proximity Circle and the remaining game duration $\trm(\thetaExit)$. 
    From the unconstrained optimal control analysis in \Cref{sec:optimal_control}, the optimal heading angle at $\x_{\rm{exit}}$ is $\psiExit = \atan(\sin\thetaExit,\ \cos\thetaExit - \trm(\thetaExit))$. 
    This heading angle results in the relative heading $\phiExit$ in the Pursuer-fixed frame, which can be computed from \Cref{eq:given_psi}:
    \begin{align*}
        \phiExit = \atan(\mu\sin\psiExit,\ \mu\cos\psiExit - 1). 
    \end{align*}
    Given \eqref{eq:exitCondition}, we conclude that $\phiExit = \thetaExit + \tfrac{\pi}{2}$. 
    Consequently, the straight line trajectory with heading $\psiExit$ does not intersect with the Proximity Circle.
    This implies that the optimal trajectory starting from $\x_{\rm{exit}}$ with remaining duration $\trm(\thetaExit)$ is indeed the straight line trajectory with heading $\psiExit$.
    %
    %
    % It would be suboptimal for $E$ to continue riding the Proximity Circle at $\thetaExit$ instead of exiting the circle.
    % This is because the trajectory continuing beyond $\thetaExit$ and then departing the Proximity Circle is not time-optimal as it is non-straight (due to~\cite{zermelo1931navigationsproblem}) and thus, since both of these possible trajectories start from $\lVert \x \rVert = 1$ the former can get further away.
    % Also, exiting the Proximity Circle before $\thetaExit$ with a heading that satisfies~\Cref{eq:ψ_xy_compact} is infeasible as $E$ would immediately be captured.
    % \todo{To be completed. } 
    % [Sketch] The heading direction in the Pursuer fixed frame corresponding to $\psiExit$ is tangent to the Proximity Circle. That is, $\phi(\psiExit) = \thetaExit + \frac{\pi}{2}$, where $\phi(\psiExit)$ is obtained from \eqref{eq:given_psi}.
\end{proof}

After substituting \eqref{eq:finalUnconstrained}, equation \eqref{eq:exitCondition} can be simplified to
\begin{align}
    \label{eq:exit_time_angle}
     \left( \mu \trm + \sqrt{1 + \trm^2 - 2\trm \cos \thetaExit} \right) \cos \thetaExit - \mu = 0. 
\end{align}
where we have suppressed the argument of $\theta$ in $\trm$ for brevity.

\begin{proposition} \label{prop:thetaExit}
    For any given $\trm > 0$, \eqref{eq:exit_time_angle} has a unique solution $\thetaExit$ in the range $[\cos^{-1} μ, \frac{\pi}{2}]$.
    Furthermore, as $t_r \to \infty$, we have $\thetaExit \to \tfrac{\pi}{2}$.
\end{proposition}

\begin{proof}
    The proof is presented in~\Cref{AP:thetaExit}.
\end{proof}

\subsection{Full Constrained Solution}\label{sec:fullSolution}

In summary, trajectories which activate the constraint and admit survival on the constraint begin with a unconstrained portion in the direction of $\phi_{\tan}$ which hits the constraint at location $\thetaEntry$ and time $t_{\tan}$. Then, the trajectory rides along the constraint until reaching $\thetaExit$. Finally, it leaves the constraint on the new unconstrained trajectory described by~\eqref{eq:finalUnconstrained}. We now put all three pieces together.

Equation \eqref{eq:exit_time_angle} dictates the relationship between the exit angle and the remaining time. 
Solving \eqref{eq:exit_time_angle} for $\trm$, one obtains 
\begin{align} \label{eq:putting_all}
    \trm = \frac{\sin\thetaExit\sqrt{\mu^2 - \cos^2\thetaExit} - (\mu^2 - \cos^2\thetaExit)}{(1-\mu^2)\cos\thetaExit}. 
\end{align}
Recall from \eqref{eq:t_remain} that $\trm$ denotes the remaining duration of the game. 
From \eqref{eq:t_s} we obtain the time spent to reach the location $\thetaExit$ on the Proximity Circle.
Therefore, using \eqref{eq:t_remain} and \eqref{eq:t_s}, we may re-write \eqref{eq:putting_all} as 
\begin{align} \label{eq:theta_exit}
\begin{split}
    \frac{\sin\thetaExit\sqrt{\mu^2 - \cos^2\thetaExit} - (\mu^2 - \cos^2\thetaExit)}{(1-\mu^2) \cos\thetaExit}  \\
     +~~ \tarc(\thetaEntry, \thetaExit)= T - t_{\tan}.
\end{split}
\end{align}
Given initial location $\x_0$, one obtains unique $t_{\tan}$ and $\thetaEntry$ from \eqref{eq:t_tan} and \eqref{eq:thetaTan}.
One may then solve \eqref{eq:theta_exit} to find the optimal exit location for a given initial location for $E$.
In general, \eqref{eq:theta_exit} needs to be solved numerically, however, we may employ a bisection search method to find $\thetaExit$ since the LHS of \eqref{eq:theta_exit} is a monotonic function of $\thetaExit$.
To show this, let us denote the LHS of \eqref{eq:theta_exit} by $f(\thetaExit)$, and obtain
%\dm{we should independently verify this} 
\begin{align*}
   & \frac{\mathrm{d} f(\thetaExit)}{\mathrm{d}\thetaExit} = \\ &\frac{\mu^2\sin\thetaExit}{\cos^2\thetaExit\sqrt{\mu^2 - \cos^2\thetaExit}(\sin\thetaExit + \sqrt{\mu^2 - \cos^2\thetaExit})},
\end{align*}
which is strictly positive in the domain $\thetaExit \in [\cos^{-1}\mu , \pi/2]$, and hence, $f(\thetaExit)$ is a monotonically increasing function in $[\cos^{-1}\mu , \pi/2]$.

The final location of $E$ in the Pursuer-fixed frame is at 
\begin{align}
    \begin{bmatrix}
        x_f \\ y_f
    \end{bmatrix} =
    \begin{bmatrix}
        \cos\thetaExit - v(\frac{\pi}{2}+\thetaExit)\trm(\thetaExit) \sin\thetaExit \\
        \sin\thetaExit + v(\frac{\pi}{2}+\thetaExit)\trm(\thetaExit) \cos\thetaExit
    \end{bmatrix},
\end{align}
where $v(\frac{\pi}{2}+\thetaExit)$ is obtained from the first case of \eqref{eq:given_phi}, i.e., $v(\frac{\pi}{2}+\thetaExit)= \sin\thetaExit + \sqrt{\mu^2 - \cos^2\thetaExit}$.
Consequently, $E$'s optimal final distance from $P$ is
\begin{align*}
    d_f^*(T, \x_0) \triangleq \|\x_f\|=  \sqrt{1 + \trm(\thetaExit)^2 v(\tfrac{\pi}{2}+\thetaExit)^2}.
\end{align*}
Using \eqref{eq:putting_all} to substitute $\trm(\thetaExit)$ and using $v(\frac{\pi}{2}+\thetaExit)= \sin\thetaExit + \sqrt{\mu^2 - \cos^2\thetaExit}$, one may verify that
\begin{align*}
    \trm(\thetaExit) v(\tfrac{\pi}{2}+\thetaExit) = \frac{\sqrt{\mu^2 -\cos^2\thetaExit}}{\cos\thetaExit}.
\end{align*}
Consequently, $d_f^*$ simplifies to 
\begin{align} \label{eq:simplified_final_distance}
    d_f^*(T, \x_0) = \frac{\mu}{\cos\thetaExit},
\end{align}
where $\thetaExit$ is the unique solution to \eqref{eq:theta_exit}.
In \eqref{eq:simplified_final_distance}, $\thetaExit$ depends on $T$ and $\x_0$. 
The dependence on $\x_0$ is via the terms $t_{\tan}$ and $\thetaEntry$ in \eqref{eq:theta_exit}.

\begin{remark}
    One may trivially verify that $d_f^*(T, \x_0) > 1$ since $\thetaExit \in (\cos^{-1}\mu, \tfrac{\pi}{2})$. 
    Furthermore, as $T \to \infty$, we obtain $\thetaExit \to \tfrac{\pi}{2}$ from \eqref{eq:theta_exit} and consequently, $d_f^*(T,\x_0) \to \infty$. 
\end{remark}

\begin{lemma}
    \label{lem:T_critical}
    Given an initial position, $\x_0$, the critical setting of $T$ (denoted $T_c$) which partitions the parameter space into unconstrained optimal trajectories (for $T \leq T_c$) and constrained optimal trajectories (for $T > T_c$) is given by
    \begin{equation}
        \label{eq:T_critical}
        T_c(\x_0) = t_{\tan} + \cos \theta_{\tan} - \frac{\sin \theta_{\tan}}{\tan \psi_{\tan}},
    \end{equation}
    where $t_{\tan}$, $\theta_{\tan}$, and $\psi_{\tan}$ are defined by~\Cref{eq:t_tan}, \Cref{eq:thetaTan}, and~\Cref{eq:phi_tan} substituted into~\Cref{eq:given_phi}, respectively.
\end{lemma}
\begin{proof}
    The critical trajectory grazes the Proximity Circle, when the optimal unconstrained $ψ^*$ is equal to the heading required to reach the tangent point $θ_{\tan}$, as given by Section~\ref{sec:enteringConstraint}. % , i.e., $ψ_{\tan}$ (which is obtained by substituting~\Cref{eq:phi_tan} into~\Cref{eq:given_phi}).
    For $T > T_c$, the optimal unconstrained heading, $ψ^*$, from~\Cref{eq:ψ_xy_compact} continues to increase towards $\pi$ which causes the unconstrained trajectory to cross into the Proximity Circle.
    For decreasing $T < T_c$ the optimal unconstrained heading decreases towards $0$ resulting in unconstrained trajectories which remain feasible throughout the duration.
    %Concerning the case when $T = T_c$, the requirement is that 
    %This equality must hold along the entire trajectory (since $ψ^*$ is constant per~\Cref{lem:optimal_control}), but, in particular, it must hold when $E$ reaches $θ_{\tan}$.
    Substituting $\x(t_{\tan}) = \begin{bmatrix} \cos θ_{\tan} & \sin θ_{\tan} \end{bmatrix}\T$ into~\Cref{eq:ψ_xy_compact} and setting equal to $ψ_{\tan}$ gives
    \begin{equation*}
        \atan\left(\sin θ_{\tan},\ \cos θ_{\tan} - t_r\right) = ψ_{\tan}.
    \end{equation*}
    Taking the $\tan$ of the above expression allows solving for the requisite time remaining $\trm$ which ensures that the $ψ^*(t_{\tan}) = ψ_{\tan}$.
    The total time, then, is this $\trm$ summed with the time needed to reach $θ_{\tan}$ which is simply $t_{\tan}$ from~\Cref{eq:t_tan}, which results in~\Cref{eq:T_critical}.
\end{proof}

\begin{figure}
    \centering
    \includegraphics[trim = 140 330 140 330, clip, width = \linewidth]{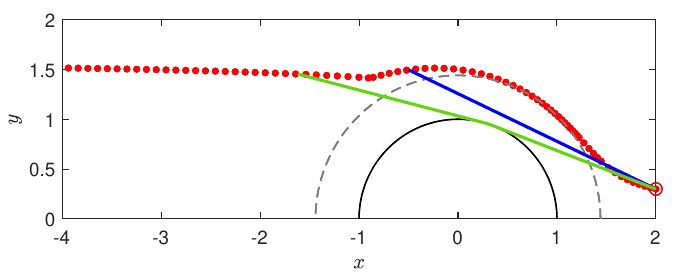}
    \put(-130, 81) {$T = T_c$}
    \put(-117, 74) {$\downarrow$}
    \put(-190, 30) {\boxed{\mu = 0.6}}
    \caption{ For $\x_0=\begin{bmatrix}2 & 0.3\end{bmatrix}\T$, the final locations of $E$ (in the Pursuer-fixed frame) is plotted using red dots as $T$ is varied from $0$ to $4$. 
    For $T=2.1$ and $T=2.6$ the corresponding optimal trajectories of $E$ are plotted in blue and green, respectively.  $T_c\approx 2.19$ is the critical time beyond which $E$'s optimal trajectory becomes constrained.
    The dashed black circle has a radius of $\min_T d_f^*(T,\x_0)$.}
    \label{fig:nonzeroInitialY}
\end{figure}

\begin{figure}
    \centering
    \includegraphics[trim = 140 330 140 330, clip, width = \linewidth]{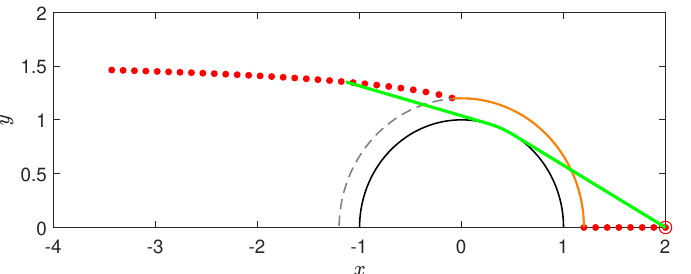}
    \put(-190, 30) {\boxed{\mu = 0.6}}
    \caption{For $\x_0=[2~~0]\T$, the final locations of $E$ (in the Pursuer-fixed frame) is plotted using red dots as $T$ is varied from $0$ to $4$. 
    At time $T=x_0 =2$, any point on the orange arc is an optimal final location for $E$.
    $E$'s optimal trajectory corresponding to $T=2.6$ is plotted in green. 
    }
    \label{fig:Ef_from_y0}
\end{figure}

\Cref{fig:Ef_from_y0,fig:nonzeroInitialY} depict example trajectories and final locations with varying final times $T$.
The final locations $\x_f$ changes smoothly with $T$ except at the single point $T = T_c$ (see \Cref{fig:nonzeroInitialY}). 
This is because all the trajectories are unconstrained for $T< T_c$ (e.g., the blue trajectory) and constrained for $T> T_c$ (e.g., the green trajectory). 

In \Cref{fig:Ef_from_y0} we consider a special initial condition where $y_0 = 0$. 
For all $T < x_0$, the optimal heading is $\psi^* = 0$ and the optimal trajectory is unconstrained, whereas for all $T> x_0$, the optimal trajectory is constrained and has three phases (e.g., the green trajectory). 
At $T= x_0$, any constant $\psi$ that does not satisfy \Cref{lm:constraintActive} is optimal. 
The possible final locations $\x_f$ corresponding to these $\psi$'s are denoted by an orange arc in \Cref{fig:Ef_from_y0}.
It is interesting that, whereas the optimal final location is unique for each $T \ne x_0$, we have a continuum of optimal final locations for $T= x_0$.

Of particular note is that none of the optimal trajectories terminate on the Proximity Circle.
In fact, for any $T$, the final location never enters the black dashed circles.   
This observation motivates the question of, for a given initial condition, what selection of $T$ minimizes the final distance, and what is the corresponding minimum distance?
This is analyzed via a game formulation discussed in the next section.

\section{Nash Equilibrium}

Consider a game in the $(\psi(t), T)$ space where $E$ picks a trajectory described by $\psi(t)$ and $P$ picks the final time $T$.
The previous sections described $E$'s best response to a particular choice of $T$, so it is natural to consider a Stackleberg game where $P$ selects and announces his choice of $T$, knowing that $E$ will play best-response in turn. 
We will describe the Stackleberg equilibrium and then observe that it is, in fact, also a Nash equilibrium.%to optimize the final distance $d_f^*(T, \x)$. 

Given an initial location $\x_0$ for $E$ and an announced final time $T$, let $d_f^*(T, \x_0)$ denote the final distance between $P$ and $E$ provided by $E$'s best response in the previous sections.
Let $T_{\min}(\x_0)$ denote the Stackleberg solution such that
\begin{align} 
    T_{\min}(\x_0) = \argmin_{T>0} d_f^*(T, \x_0).  
\end{align}

\begin{lemma} \label{lem:optimalT}
    For any $\x_0 \notin \Omega_{\text{capture}}$, 
    \begin{align}
    \label{eq:Tmin}
    T_{\min}(\x_0) = \begin{cases}
        x_0 - \tfrac{\mu}{\sqrt{1-\mu^2}}y_0, &\text{ if } x_0 - \tfrac{\mu}{\sqrt{1-\mu^2}}y_0 > 0,\\
        0, &\text{  otherwise,}
    \end{cases}
\end{align}
and consequently, 
\begin{align}
    d_f^*(T_{\min}, \x_0) = \begin{cases}
        \mu x_0 + \sqrt{1-\mu^2}y_0, &\text{ if } T_{\min}(\x_0) > 0,\\
        \|\x_0\|, &\text{  otherwise.}
    \end{cases}
\end{align}
\end{lemma}
\begin{proof}
    The proof is presented in \Cref{AP:optimalT}.
\end{proof}

For $T= T_{\min}$, $E$'s optimal heading from \eqref{eq:ψ_xy_compact} yields
\begin{align}\label{eq:nashEvader}
    \psi(t,\x) = \cos^{-1}\mu,
\end{align}
which is \textit{independent} of the final time or the initial/current location.
%and the final time $T$. 
This heading angle results in $\phi = \pi - \sin^{-1}\mu$ in the Pursuer-fixed frame, which happens to be the maximum possible value of $\phi$. Notice that while $E$ can choose any admissible trajectory and we have made no assumption forcing its selection to be only a straight-line trajectory, this heading will not intersect the constraint (from Remark~\ref{rem:nashUnconstrainted}), and so the entire trajectory will be one single straight line.

% \begin{remark}
%     The level sets $T_{\min}(\x_0) = c$ are straight lines with slope $\cos^{-1}\mu$ for all $c > 0$. 
%     These level sets are perpendicular to the level sets of $d_f^*(T_{\min}, \x_0)$, which are also straight lines with slope $\pi - \sin^{-1} \mu$.
% \end{remark}

Finally, the following theorem provides Nash equilibrium strategies for this game. \begin{theorem} \label{thm:minmaxStrategy}
    For a given $\x$, an equilibrium pair of policies is 
    \begin{align}
        (\psi_{\rm{NE}}, T_{\rm{NE}}) = \left(\cos^{-1}\mu,~ \max\left\{ x - \tfrac{\mu}{\sqrt{1-\mu^2}}y, 0 \right\}\right).
    \end{align}
\end{theorem}
\begin{proof}
    $\psi_{\rm{NE}}$ is the best response to $T_{\rm{NE}}$ as per~\eqref{eq:nashEvader}.
    
    Now, given $\psi_{\rm{NE}}$, $E$'s location at time $t$ is 
    \begin{align*}
        \x(t) = \begin{bmatrix}
            x + \left(\mu^2 - 1\right)t \\
            y + \mu\sqrt{1-\mu^2} t
        \end{bmatrix}.
    \end{align*}
    Consequently, the distance between $E$ and $P$ at time $t$ is
    \begin{align*}
        %d(t) = \sqrt{x^2 + y^2}  + (1-\mu^2) (t^2 - 2tT^*).
        d(t) = \sqrt{x^2 + y^2  + (1-\mu^2) \Big(t^2 - 2t\big(x - \tfrac{\mu}{\sqrt{1-\mu^2}}y\big)\Big)}.
    \end{align*}
    Consequently, the minimum value of $d(t)$ is attained at $t=T_{\rm{NE}}$.
    %if $T^*\ge 0$, or at $t=0$ when $T^*<0$. 
\end{proof}

    For the special case of $y=0$ discussed in Section~\ref{sec:fullSolution}, $T_{\rm{NE}} = x$. 
    As observed before for this special case, any constant heading angle is an optimal best-response to $T_{\rm{NE}}$ as long as that heading does not intersect the Proximity Circle.
    However, only the $\psi$ given by~\eqref{eq:nashEvader} is non-exploitable.
    Moreover, there is a significant jump in optimal best-response for small perturbations away from $T_{\rm{NE}}$.
    For $y=0$ and $T < x$, the optimal heading, as per \eqref{eq:ψ_xy_compact}, is to pick $\psi = 0$, which is a pure evasion strategy. 
    On the other hand, when $y=0$ and $T>x$, the optimal strategy is to reach the Proximity Circle, ride the circle, and leave the circle tangentially at $\thetaExit$.
    
\begin{remark}
    Our analysis requires $P$ and $E$ to pre-select policies at the beginning of the game. 
    In general, such policies could be closed-loop feedback policies~\cite{dorothy2021one}. 
    Here, $P$ would obtain no additional information at any time in the game and would have no basis on which to have a feedback policy rather than an open-loop one.
    However, $E$'s situation is different. 
    The Nash equilibrium policy requires $E$ to commit to its open loop policy until time $T_{\rm{NE}}$, expecting the game to end there in equilibrium. 
    Hence, $E$ cannot gain any information that would allow him to exploit any suboptimal behavior by $P$ during the time $t\leq T_{\rm{NE}}$.
    However, if $P$ plays suboptimally with $T>T_{\rm{NE}}$, there remains an open question as to how $E$ should respond after $T_{\rm{NE}}$ in light of the new information.
    Moreover, in the special case of $y=0$, due to the jump between the optimal headings for $T<T_{\rm{NE}}$ and $T>T_{\rm{NE}}$, $E$ may find itself having significant regret for having committed to the Nash solution up to time $T_{\rm{NE}}$.
    % If an Evader located at $[x~~ 0]\T$ is not aware of the choice of $T$ made by the Pursuer, then it faces a dilemma since there is a jump between the optimal headings for $T<x$ and $T>x$. 
    % The Evader has to commit to a strategy first.
    A similar situation has also been observed in Section~IV-C of \cite{maity2023efficient}. 
\end{remark}

% \subsection{Evader's Dilemma}
% If the Evader is not aware of the choice of $T$ made by the Pursuer, then it faces a dilemma. 
% If $T=T^*$, then \Cref{thm:minmaxStrategy} provides the optimal heading angle.
% On the other hand, if $T < T^*$, then a pure evasion strategy

% \subsection{Maximum Survival Time when Capture is Guaranteed}\label{sec:maxSurvive}

% \begin{lemma}
%     If the Evader's initial location satisfies 
%     \begin{align} \label{eq:no_escape_zone}
%         \mu x + \sqrt{1- \mu^2}|y| < 1,\quad x \ge \mu
%     \end{align}
%     the Evader will be eventually captured if the game duration is sufficiently long.
% \end{lemma}

% For an Evader's initial location $\x$ satisfying \eqref{eq:no_escape_zone}, the maximum \textit{survival time} is 
% \begin{align}
% \begin{split}
%     &T(\x) = \\
%     &\frac{1}{1-\mu^2}\left[ (x-\mu) - \sqrt{(1-\mu x)^2 - (1-\mu^2)y^2} \right]. 
%     \end{split}
% \end{align}
% Alternatively, for a given duration $T$, the set of initial locations for which the Evader is not able to survive is given by the set:
% \begin{align*}
%     \begin{multlined}
%     \Omega(T) = \left\{\x~|~ (x - T)^2 + y^2 < (1-\mu T)^2, \right. \\
%     \left. \text{and  \eqref{eq:no_escape_zone} is satisfied}\right\}.
%     \end{multlined}
% \end{align*}

    \section{Conclusion }
In this work, we studied the problem of optimal evasion from a sensing limited Pursuer that moves on a straight line.
We identified the initial locations for the Evader (i.e., the set $\Omega_{ \text{capture} }$) for which capture is guaranteed, and for such initial conditions, we derived the evasion strategy that maximizes the evader's survival time.
When $\x_0 \notin \Omega_{ \text{capture} }$, the optimal evasion strategy is either to move along a constant heading (i.e., when the conditions in Lemma~\ref{lm:constraintActive} are not satisfied) or it consists of three segments: a straight line trajectory to reach the Proximity Circle, an arc to glide on the Proximity Circle, and a straight line segment leaving  the Proximity Circle. 

We discussed the case where the pursuer can pick the final time $T$, and in that case, the minmax equilibrium strategy resulted into a constant heading trajectory for the evader, where the heading is independent of the evader's initial location.
These results form a basis for addressing the full intermittent sensing pursuit-evasion scenario in future work.

\bibliographystyle{ieeetr}
\bibliography{arXiv}

\appendix 

\subsection{Proof of \Cref{lem:guaranteedCapture}} \label{AP:guaranteedCapture}
For a given time $t$ and an initial location $\x_0$, the set 
\begin{equation*}
    {\mathcal{R}}(t, \x_0) = \{\x_f ~|~ (x_f - (x_0 -t))^2 + (y_f-y_0)^2 \le \mu^2 t^2  \}
\end{equation*}
contains all the possible locations of $E$ at time $t$ in the Pursuer-fixed frame. 
The point
\begin{align}
    \x^* = \Big(1 + \tfrac{\mu t}{\sqrt{(x_0 - t)^2 + y_0^2}}\Big) \begin{bmatrix}
        x_0 - t \\ y_0
    \end{bmatrix}
\end{align}
belongs to $\mathcal{R}(t, \x_0)$ and is the farthest point from the origin of the Pursuer-fixed frame.
Let us define
\begin{align*}
    \bar{d}_f(t, \x_0) \triangleq \|\x^*\| =  \sqrt{(x_0 - t)^2 + y_0^2} + \mu t.
\end{align*} 
For a given $\x_0$, if there exists a $t>0$ such that $\bar{d}_f(t, \x_0)< 1$, then $E$ will enter the interior of the Proximity Circle, regardless of its strategy, by time $t$. 
The condition $\bar{d}_f(t, \x_0)< 1$ yields the inequality:
\begin{align} \label{eq:inequalityforT}
    (1-\mu^2)t^2 - 2(x_0 - \mu) t + \|\x_0\|^2 -1  < 0.
\end{align}
% Equation~\Cref{eq:inequalityforT} holds true for $T$ values in the range: 
%  \begin{align*}
%      \frac{(x_0 - \mu) - \sqrt{(x_0 - \mu )^2 - (1-\mu^2) (\|\x_0\|^2 -1)}}{1-\mu^2} < T \\
%       <  \frac{(x_0 - \mu) + \sqrt{(x_0 - \mu )^2 - (1-\mu^2) (\|\x_0\|^2 -1)}}{1-\mu^2}.
%  \end{align*}
%  \begin{align*}
%      \frac{(x_0 - \mu) - \sqrt{(1-\mu x_0)^2 - (1-\mu^2) y_0^2}}{1-\mu^2} < T \\
%       <  \frac{(x_0 - \mu) + \sqrt{(1-\mu x_0)^2 - (1-\mu^2) y_0^2}}{1-\mu^2}
%  \end{align*}
Given \eqref{eq:no_escape_zone}, $t = T_{\rm{survive}}$ is well defined and is the smaller root of the LHS of \eqref{eq:inequalityforT}. 
Therefore, for any final time $T> T_{\rm{survive}}$, there exists a $t \in (T_{\rm{survive}}, T]$ such that $\|\x(t)\| < 1$.
Thus, the given conditions in the lemma is sufficient.

% \eqref{eq:inequalityforT} holds for $t \in ( T_{\rm{survive}}, \bar{T}_{\rm{survive}} )$, where 
% \begin{align*}
    
% \end{align*}

% There exists a $T > T_{}$ satisfying \Cref{eq:inequalityforT}  if and only if, 
% \begin{align}
%     \begin{cases} \label{eq:noEscapeConditions}
%         (x_0 - \mu) > 0,  \\
%         (1-\mu x_0)^2 - (1-\mu^2) y_0^2 > 0.
%     \end{cases} 
% \end{align}
% Given that we are interested in $x_0 \ge 0$ and $y_0 \ge 0$, \eqref{eq:noEscapeConditions} simplifies to \eqref{eq:no_escape_zone}. 
% Therefore, the given conditions in the lemma is sufficient.

To show that the conditions are also necessary, let $E$ follow the direction $\psi = \cos^{-1}\mu$. 
We now show that, under this heading, capture is possible only if all the conditions in the lemma are satisfied.
To this end, we obtain from \eqref{eq:f_nondim} that $\x(t) = \x_0 + t \begin{bmatrix} (\mu^2 -1) & \mu\sqrt{1-\mu^2} \end{bmatrix}\T$, and consequently, 
\begin{align*}
    \|\x(t)\|^2 &= \|\x_0\|^2 + (1-\mu^2)t^2 - 2(1-\mu^2)t \underset{c}{\underbrace{\begin{bmatrix}
        1 & \tfrac{\mu}{\sqrt{1-\mu^2}} \end{bmatrix}\x_0 }} \\
    & = \|\x_0\|^2 - (1-\mu^2) c^2 + (1-\mu^2) (t-c)^2, \\
    & = (\mu x_0 + \sqrt{1-\mu^2}y_0)^2 + (1-\mu^2) (t-c)^2.
\end{align*}
Therefore, to ensure $\|\x(t)\|<1$ for some $t > 0$, it is necessary that (i) $\mu x_0 + \sqrt{1-\mu^2}y_0 < 1$, (ii) $c>0$. 
Combining (i) and (ii), we obtain $x_0 > \mu$. 
Thus, \eqref{eq:no_escape_zone} is necessary.
% Solving for a $t$ to satisfy $\|x(t)\|<1$, we obtain $t > T_{\rm{survive}}$. 
To show that \eqref{eq:max_survival_time} is also necessary, we proceed as follows. 
From the expression of $T_{\rm{survive}}$, one may verify that 
\begin{align*}
    \left( \frac{x_0 - T_{\rm{survive}}}{1-\mu T_{\rm{survive}}}\right)^2 +
    \left( \frac{y_0 }{1-\mu T_{\rm{survive}}}\right)^2 = 1.
\end{align*}
For an initial location $\x_0$, consider the heading angle $\psi$ such that $\cos\psi = \tfrac{x-T_{\rm{survive}}}{1 - \mu T_{\rm{survive}}}$ and $\sin\psi = \tfrac{y_0}{T_{\rm{survive}}}$. 
Using \eqref{eq:f_nondim}, we obtain that, for any $t$, 
\begin{align} \label{eq:x(t)}
    \x(t) = \begin{bmatrix}
        x_0 + t \tfrac{\mu x_0 - 1}{1- T_{\rm{survive}}} \\
        y_0 + t \tfrac{y_0}{T_{\rm{survive}}}
    \end{bmatrix}.
\end{align}
Solving for $t$ such that $\|\x(t)\| = 1$, we obtain two solutions (say $t_1$ and $t_2$). 
For $\x_0$'s satisfying \eqref{eq:no_escape_zone}, both $t_1$ and $t_2$ are real and non-negative. 
Without loss of generality, we assume $t_1 < t_2$.
Furthermore, one may verify that $\|\x(T_{\rm{survive}})\| = 1$ and $t_1 = T_{\rm{survive}}$. 
Therefore, to ensure $\|\x(t)\|< 1$ for some $t$ we must have $T > T_{\rm{survive}}$. 
Thus, \eqref{eq:max_survival_time} is necessary. \hfill $\blacksquare$

% Recall, that, $T_{\rm{survive}}$ is a root of the LHS of the \eqref{eq:inequalityforT}.
% The LHS of \eqref{eq:inequalityforT} can be rearranged as $(x_0 - T)^2 + y_0^2 - (1-\mu T)^2 $.

% For an $\x_0$ satisfying \eqref{eq:no_escape_zone} (equivalently, \eqref{eq:noEscapeConditions}), one may verify that if $T$ satisfies \eqref{eq:max_survival_time}, then there exists a $t \in [T_{\rm{survive}}, T]$ such that $\|\x(t)\| < 1$, i.e., capture happens.

\subsection{Proof of \Cref{prop:thetaExit}} \label{AP:thetaExit}
Equation \Cref{eq:exit_time_angle} can be rearranged as 
\begin{align} \label{eq:UniquethetaExit}
    \cos \thetaExit \sqrt{1 + \trm^2 - 2\trm \cos \thetaExit}    = \mu(1- \trm \cos \thetaExit).
\end{align}
Taking squares on both sides of the last equation and rearranging the terms yields, 
\begin{align*}
    2\trm \cos^3\thetaExit - (1 + &\trm^2(1-\mu^2))\cos^2\thetaExit \\
     &- 2\mu^2\trm\cos\thetaExit + \mu^2 = 0.
\end{align*}
The last expression is a cubic polynomial in $\cos\thetaExit$. 
Let us denote this cubic polynomial by $f(s) = 2\trm s^3 - (1 + \trm^2 (1-\mu^2))s^2 - 2\mu^2\trm s + \mu^2$. 
Notice that $f(-1)  < 0$, $f(\mu) < 0$ and $f(0) > 0$ for all $\trm >0$ and $\mu \in (0, 1)$. 
Therefore, the three roots of $f(s)$ are within the ranges $(-1, 0)$, $(0, \mu)$ and $(\mu, \infty)$, respectively.   
Thus, \eqref{eq:UniquethetaExit} has a unique solution for $\thetaExit$ in the range $[\cos^{-1}\mu, \tfrac{\pi}{2}]$. 

Let $r_1, r_2,$ and $r_3$ be the roots of the polynomial $f(s)$. 
Then,
\begin{align*}
   & r_1 + r_2 + r_3  = \frac{1 + \trm^2(1-\mu^2)}{2\trm}, \quad r_1r_2 r_3 = - \frac{\mu^2}{2\trm}, \\
   &r_1r_2 + r_2r_3 + r_3 r_1 = - \mu^2.
\end{align*}
Therefore, as $\trm\to \infty$, two of the roots approach $0$ and the third one to $\infty$. 
Consequently, as $\trm\to \infty$, we have $\thetaExit \to \tfrac{\pi}{2}$. \hfill $\blacksquare$

\subsection{Proof of \Cref{lem:optimalT}} \label{AP:optimalT}

For a given a final time $T$ and an initial location $\x_0 \notin \Omega_{ \text{capture} }$, the possible final locations of $E$ in the Pursuer-fixed frame is contained in the set
\begin{equation*}
    {\mathcal{R}}(T, \x_0) = \{\x_f ~|~ (x_f - (x_0 -T))^2 + (y_f-y_0)^2 \le \mu^2 T^2  \}. 
\end{equation*}
The farthest point on this set from the origin of the Pursuer-fixed frame is at a distance of 
\begin{align*}
    \bar{d}_f(T, \x_0) \triangleq \sqrt{(x_0 - T)^2 + y_0^2} + \mu T.
\end{align*} 
Since, for any given pair $(T, \x_0)$, $\mathcal{R}(T, \x_0)$ is an over-approximation of the $E$'s reachable region, we may write  $d_f^*(T, \x_0) \le \bar{d}_f(T, \x_0)$, where $d_f^*(T, \x_0)$ is the maximum distance obtained by $E$'s optimal strategy. 

Minimizing $\bar{d}_f (T, \x_0)$ w.r.t. $T$ for a given $\x_0$ yields the $T_{\min}$ given in \eqref{eq:Tmin}.
Next, we verify that $T_{\min}$ is also the minimizer of $d_f^*$. 
To that end, we obtain from \eqref{eq:ψ_xy_compact} that $\psi^* = \cos^{-1}\mu$ when $T= T_{\min}.$ 
Using \Cref{lm:constraintActive}, we further notice that, for all $\x_0 \notin \Omega_{ \text{capture} }$,  $E$'s straight line trajectory with $\psi = \cos^{-1}\mu$ does not intersect with the Proximity Circle. 
Therefore, a straight line trajectory with $\psi = \cos^{-1} \mu$ is feasible for \textit{any} $T$, and the final distance resulting from this heading is
\begin{align}
\begin{split}
    \sqrt{\|\x_0\|^2 + (1-\mu^2)(T^2 - 2(x_0- \frac{\mu}{\sqrt{1-\mu^2}}y_0)T)} \\
    \qquad\triangleq d_f^{\rm{subopt}}(T, \x_0).
    \end{split}
\end{align}
Since the straight line trajectory with heading $\cos^{-1}\mu$ is not necessarily the optimal strategy, we may write 
\begin{align*}
    d_f^*(T, \x_0) &\ge d_f^{\rm{subopt}}(T, \x_0) \ge d_f^{\rm{subopt}}(T_{\min}, \x_0),
\end{align*}
where the last equality holds only at $T= T_{\min}$. Therefore, $d_f^*(T, \x_0)$ is lower bounded by the constant $d_f^{\rm{subopt}}(T_{\min}, \x_0)$ and upper bounded by $\bar{d}_f(T, \x_0)$ for all $T$. 
Since $\bar{d}_f(T_{\min}, \x_0) = d_f^{\rm{subopt}}(T_{\min}, \x_0)$, we conclude that $T_{\min}$ is the unique minimize of $d_f^*(T, \x_0)$. \hfill $\blacksquare$

\end{document}